\documentclass{scrartcl}
\usepackage{authblk}
\usepackage{amsmath}
\usepackage{amssymb}
\usepackage{amsthm}
\usepackage{csquotes}
\usepackage[british]{babel}
\usepackage[
    giveninits,
    backend=biber,
    sorting=none,
    citestyle=numeric-comp,
    maxnames=10,
    style=vancouver,
]{biblatex}
\usepackage{xr-hyper}
\usepackage{hyperref}
\usepackage{tabularx}
\usepackage{booktabs}
\usepackage{cleveref}
\usepackage{fmtcount}
\usepackage{environ}
\usepackage{graphicx}
\DeclareGraphicsExtensions{.pdf}

\usepackage[usenames]{xcolor}
\definecolor{mpl0}{HTML}{1f77b4}

\hypersetup{
    colorlinks=true,
    linkcolor=mpl0,
    citecolor=mpl0,
    urlcolor=mpl0,
}

\newcommand{\expit}{\sigma}
\newcommand{\transpose}[1]{#1^\intercal}
\DeclareMathOperator{\logit}{logit}
\DeclareMathOperator{\logp}{log1p}
\DeclareMathOperator{\logsumexp}{logsumexp}

\newcommand{\abs}[1]{\left|#1\right|}
\newcommand{\card}[1]{\left|#1\right|}
\newcommand{\population}{N}
\newcommand{\seeds}{U}
\newcommand{\blauspace}{\mathbb{B}}
\newcommand{\twocol}[1]{\multicolumn{2}{>{\hsize=\dimexpr2\hsize+3\tabcolsep+2\arrayrulewidth\relax}X}{#1}}

\newcommand{\titlecaption}[2]{\caption[#1]{\emph{#1} #2}}
\newtheorem{prop}{Proposition}
\crefname{prop}{proposition}{propositions}

\title{Inference of a universal social scale and segregation measures using social connectivity kernels}
\author{Till Hoffmann}
\author{Nick S. Jones}
\affil{Department of Mathematics, Imperial College London}
\date{}

\begin{document}

\maketitle

\begin{abstract} 
    How people connect with one another is a fundamental question in the social sciences, and the resulting social networks can have a profound impact on our daily lives. Blau offered a powerful explanation: people connect with one another based on their positions in a social space. Yet a principled measure of social distance, allowing comparison within and between societies, remains elusive.

    We use the connectivity kernel of conditionally-independent edge models to develop a family of segregation statistics with desirable properties: they offer an intuitive and universal characteristic scale on social space (facilitating comparison across datasets and societies), are applicable to multivariate and mixed node attributes, and capture segregation at the level of individuals, pairs of individuals, and society as a whole. We show that the segregation statistics can induce a metric on Blau space (a space spanned by the attributes of the members of society) and provide maps of two societies.

    Under a Bayesian paradigm, we infer the parameters of the connectivity kernel from eleven ego-network datasets collected in four surveys in the United Kingdom and United States. The importance of different dimensions of Blau space is similar across time and location, suggesting a macroscopically stable social fabric. Physical separation and age differences have the most significant impact on segregation within friendship networks with implications for intergenerational mixing and isolation in later stages of life.
\end{abstract}

\begin{refsection}

\section{Introduction\label{sec:introduction}}

Peter Blau proposed that individuals connect with one another based on their positions in a high-dimensional space~\cite{Blau1977}, e.g.\ a space spanned by demographic attributes. With an accrual of large-scale survey data, we now have access to demographic and relational information internationally, for whole societies, and over time~\cite{McPherson2001,McPherson2006,Mossong2008}. Despite this wealth of data, we lack two important quantities: (a)~a natural notion of distance in this space, allowing us to determine how far apart individuals are in society, and (b)~a universal characteristic scale, allowing the distance between pairs of individuals in one society to be compared to another society with different social dimensions. We will argue that the probability of forming a friendship, intimately related to work on homophily by Blau and his successors~\cite{McPherson2001}, offers both a universal characteristicscale and a notion of distance.

Homophily, the tendency for people to connect with others who are alike, is one of the most robust observations of the social sciences and shapes how our society is connected~\cite{McPherson2001}. Quantifying homophily is not only important for understanding why social ties form between some people yet not between others, but the manifestation of homophily as poorly-connected social networks can have a significant impact on dynamics unfolding upon them~\cite{Golub2012}. For example, users of online social networks, such as Facebook and Twitter, tend to connect with others who hold similar political views~\cite{Boutyline2017}. They are more likely to be exposed to information that confirms rather than challenges their beliefs~\cite{Bakshy2015}. An ``echo chamber'' effect ensues, leading to polarised opinions~\cite{DeMarzo2003}. Homophily can also have a detrimental impact on public health: clusters of individuals who mutually reinforce their belief that vaccinations are harmful can raise the likelihood of significant disease outbreaks~\cite{Salathe2008}---even if the vaccination rate is above herd immunity levels on average.

Homophily can be observed in friendships~\cite{Currarini2009, Hipp2009}, networks of discussion partners~\cite{McPherson2006}, communication networks~\cite{Wang2013, Leo2016}, marital ties~\cite{Blau1984}, and online social networks~\cite{Chang2010}. Relationships are homogeneous with respect to a wide range of attributes, including age~\cite{Marsden1988,Smith2014}, sex~\cite{Smith2014}, ethnicity~\cite{Chang2010, Blumenstock2013, Currarini2009}, education~\cite{McPherson2006, Smith2014,Johnson1989}, occupation~\cite{Chan2004}, income~\cite{Leo2016, Wang2013, Johnson1989}, religion~\cite{Platt2012}, parental~\cite{Johnson1989} and marital status~\cite{Kalmijn2007}, political ideology~\cite{Bakshy2015, Boutyline2017}, and geographical location~\cite{Lambiotte2008, Expert2011, Backstrom2010, Scellato2011, Illenberger2013}.

Segregation statistics, also referred to as segregation measures, are often used to quantify homophily~\cite{Rodriguez-Moral2016, Bojanowski2014}. Many approaches are based on co-presence in organisational units such as schools~\cite{Orfield2014}, voluntary associations~\cite{Popielarz1999}, occupations~\cite{Charles1995}, or census tracts~\cite{Reardon2004}, and we refer to them as \emph{organisational} statistics. Typically, they compare how the distribution of demographic attributes within organisational units differs from the distribution of attributes in the general population. Whilst organisational statistics are applicable whenever data can be stratified according to a variable of interest, they cannot capture segregation at smaller scales than the strata~\cite{Blumenstock2013}. For example, the ethnic composition of a set of schools may be representative of the general population, indicating that there is no (organisational) segregation. But the social networks \emph{within} schools often exhibit strong ethnic homophily~\cite{Currarini2009,Moody2001}. Organisational statistics cannot capture such social segregation.

\emph{Social} statistics of segregation, such as the assortativity coefficient~\cite{Newman2003a}, overcome these limitations by explicitly considering the interactions amongst individuals~\cite{Blumenstock2013}, but have their own difficulties: first, they usually rely on the existence of discrete groups, such as sex, ethnicity, or religion~\cite{Bojanowski2014}, and they are not applicable to continuous attributes, such as age or income. Attributes are often discretised~\cite{Lam-Morgan2012, Kalmijn2007, Kim2012}, but the boundaries between categories always suffer from some degree of arbitrariness~\cite{Reardon2004}. Second, segregation for multiple attributes can be quantified independently, but \textcite{Pelechrinis2016}, who consider a multivariate generalisation of the assortativity coefficient, note that ``a formal metric that is generally applicable'' remains elusive. Furthermore, social statistics are typically defined as summary statistics of a fully-observed social network. Consequently, we cannot easily quantify uncertainties.

In practice, the study of homophily is complicated by the scarcity of high-quality data~\cite{Butts2012, Blumenstock2013}: we need social network data together with demographic information for each person. Online social networks and the widespread use of mobile phones provide us with detailed information about connections between individuals~\cite{Golder2014}, and seemingly private traits such as socioeconomic status~\cite{Blumenstock2015, Luo2017}, sexual orientation~\cite{Wang2017}, age, gender, and political ideology can be inferred~\cite{Kosinski2013}. Unfortunately, network features are often used to predict demographic attributes~\cite{Wang2013,Blumenstock2015,Luo2017,Kosinski2013}, which would confound any study of homophily. Furthermore, ``data are [\ldots] too revealing in terms of privacy'' but, at the same time, do not provide enough information for researchers~\cite{Golder2014}. Individuals can be identified in anonymised social networks~\cite{Backstrom2011, Narayanan2008}, and augmenting the network data with demographic information would make re-identification even easier.

However, censuses and large-scale surveys collect comprehensive demographic information from respondents but usually lack data about their associates. Fortunately, some surveys have included questions about respondents' friends~\cite{Huckfeldt1983, Johnson1989}, discussion partners~\cite{Marsden1987, McPherson2006}, or support networks~\cite{Kalmijn2007, Banerjee2013}. The questions used to elicit social ties provide an imperfect observation of the immediate neighbourhood of respondents~\cite{Marin2004,Eagle2015,Eveland-Jr.2017}.

Building on the successes of conditionally-independent edge models~\cite{Snijders2011} and, in particular, latent space models for social networks~\cite{Hoff2002,Hoff2008}, we consider a generative model for social networks whose members occupy a multidimensional Blau space in \cref{sec:model}. We discuss desirable properties for social segregation statistics, and, using the generative network model, we develop a suite of statistics applicable to arbitrary attributes in \cref{sec:model-based-segregation}. The statistics capture segregation at different scales: single individuals, pairs of individuals, and society as a whole. Because of both their probabilistic foundations and their construction from the universal notion of a social tie, the segregation statistics have a universal scale, i.e.\ one unit of segregation has the same implications across different societies and at different times. We show that the segregation statistic for pairs of individuals can be a metric and can thus be used to quantify distance in Blau space. We illustrate the statistics with a simple example, and we show that it reduces to a well-known segregation statistic if the attributes are univariate and categorical: the natural logarithm of \citeauthor{Moody2001}'s $\alpha$ index~\cite{Moody2001}.

In \cref{sec:inference}, we derive the posterior for parameters of the conditionally-independent edge model given partial observations of social networks obtained from surveys. We apply our approach to nine existing datasets from the United Kingdom and two from the United States in \cref{sec:application}. Our analysis reveals that the effects of homophily on society are remarkably stable in both countries regardless of time and the specific nature of relationships. Using the suite of segregation statistics, we find that physical separation and age are the most important factors contributing to the segregation of society. In \cref{sec:discussion}, we provide recommendations for conducting surveys to infer homophily in social networks and discuss future work.

\section{Methods}

\subsection{Generative network model\label{sec:model}}

We consider a generative model for social networks for a population of $n$ individuals $\population$ who occupy a Blau space $\blauspace$ spanned by their demographic attributes, such as age, income, or sex. In contrast to common latent space models~\cite{Hoff2002, Hoff2008}, the attributes are observed, although \textcite{Hoff2002} also consider an extension including covariates. The $q$-dimensional attribute vector $x_i\in\blauspace$ for each individual $i\in \population$ is drawn independently from a distribution $P(x_i)$ of demographic attributes. Elements of the attribute vector can take continuous, ordinal, or categorical values. Connections between individuals are encoded by the binary adjacency matrix $A$ such that $A_{ij}=1$ if $j$ considers $i$ to be a friend and $A_{ij}=0$ otherwise. We assume that people do not interact with themselves such that $A_{ii}=0$ for all $i$, and that connections are undirected, although social ties need not be reciprocated in general~\cite{Ball2013}.

Given the positions of two individuals $i$ and $j$ in Blau space, we assume that connections form independently with probability $\rho(x_i, x_j)$, i.e.\ edges are conditionally-independent given the attributes of nodes~\cite{Fienberg2012}. The assumption of conditionally-in\-de\-pen\-dent edges can be problematic. For example, it is not possible to reproduce heavy-tailed degree distributions if the node density is homogeneous and the kernel is translationally invariant~\cite{Barnett2007}. Furthermore, the average degree scales linearly with the number of nodes unless the connectivity kernel $\rho$ is adjusted to compensate~\cite{Caron2017}. Nevertheless, we use conditionally-independent edge models because the connectivity kernel is intuitive, and they can capture salient features of social networks. For example, nodes in high-density regions have larger degrees on average~\cite{Barnett2007}. Similarly, members of the ethnic majority have more social ties in social networks in US high schools~\cite{Currarini2009}.

\subsection{Developing a model-based segregation statistic\label{sec:model-based-segregation}}

We have so far emphasised the desirability of metrics on social spaces with a universal characteristic scale (in the sense of being comparable between societies). After first developing universal social segregation statistics, including a notion of social separation, we will formulate both a metric and a notion of scale in \cref{sec:segregation-distance}.

In addition to addressing the challenges mentioned in \cref{sec:introduction}, a social segregation statistic should satisfy the following properties: first, the statistic should be insensitive to the overall edge density to facilitate comparison of segregation across different networks. Otherwise, the segregation statistic would depend on the size of the population because the edge density scales as $n^{-1}$ if the average degree is approximately constant. Second, following \textcite{Freeman1978}, we would like the statistic to capture the notion that segregation places ``restrictions on the access of people to one another''. Third, the statistic should be easily interpretable, and it should have a natural notion of the absence of segregation when individuals form connections without regard to their positions in Blau space. For example, the difference of within- and between-group ties considered by \textcite{Krackhardt1988} depends on the sizes of the groups even if there is no homophily: there is no natural reference point.

A single statistic cannot capture the complexities of social networks, and we develop a family of statistics applicable at different scales: (a)~the \emph{social separation} between any two individuals, (b)~the \emph{isolation} experienced by any one individual, and (c)~the \emph{social strain} experienced by society as a whole. Starting at the microscopic level, we define the \emph{social separation} between two individuals $i$ and $j$ with attributes $x$ and $y$ as the relative log odds for $j$ to connect with $i$ compared to someone who is alike: the log odds intuitively capture the probabilistic nature of the conditionally independent edge model. In particular,
\begin{align}
    \varphi(x, y) &= \logit \rho(y,y) - \logit\rho(x,y),\label{eq:social-separation}\\
    \text{where }\logit\rho &= \log\left(\frac{\rho}{1-\rho}\right)\nonumber
\end{align}
are the log odds for a connection to form with probability $\rho$~\cite{Hastie2009}. The probability $\rho(y, y)$ for $j$ to connect with someone who is alike serves as a reference point, and the statistic does not depend on the overall edge density. The social separation $\varphi$ may be understood as the isolation experienced by $i$ with attributes $x$ as a result of the behaviour of $j$ with attributes $y$. The statistic is zero if two individuals have the same demographic attributes or if they do not discriminate with respect to the attributes on which they differ. For a homophilous connectivity kernel, the statistic is positive and is a \emph{semi-metric} for Blau space; we will consider a family of connectivity kernels for which $\varphi$ is a true metric in \cref{sec:segregation-distance}.

\begin{prop}\label{prop:segregations-semi-metric}
    If the connectivity kernel is homophilous, symmetric, and the probability $\rho(x,x)$ to connect with others who are alike is independent of $x$, the social separation $\varphi$ is a \emph{semimetric}~\cite{Wilson1931}: it satisfies the properties of a metric, including non-negativity, symmetry, and the identity of indiscernables---except the triangle inequality.
\end{prop}
\begin{proof}
    First, $\varphi(x,y)\geq0$ because homophily implies that $\rho(x, y)<\rho(y, y)$ and $\logit$ is a monotonically increasing function. Second, $\varphi(x,y)=\varphi(y,x)$ because the first term of \cref{eq:social-separation} is constant by assumption and the second is symmetric because the kernel is symmetric. Third, the statistic is zero for any two individuals with the same attributes by substitution into \cref{eq:social-separation}. Similarly, if $\varphi(x,y)=0$, then $x=y$ because $\rho(x, y)<\rho(y, y)$ due to homophily.
\end{proof}

The less likely two people are to connect, the larger the social separation between them. The assumptions required for \cref{prop:segregations-semi-metric} to hold may seem restrictive, but they are satisfied by most studies of spatial networks~\cite{Barnett2007,Lambiotte2008,Butts2012,Expert2011}.

Defining social separation in terms of a generative model, i.e.\ using the connectivity kernel rather than a summary statistic of a particular dataset, provides us with two advantages: first, any uncertainty associated with inferred connectivity kernels naturally propagates to the segregation statistic, as discussed in \cref{sec:application}. Second, we can easily consider the properties of the segregation statistic under a variety of generative models without having to resort to computationally-expensive Monte Carlo simulations.

For example, consider a stochastic block model (SBM)~\cite{Snijders2011} with $K$ blocks, intra-group connection probability $\rho_\mathrm{same}$, and inter-group connection probability $\rho_\mathrm{different}<\rho_\mathrm{same}$. Substituting into \cref{eq:social-separation}, the social separation between two nodes with block membership $x$ and $y$ is
\begin{equation}
    \varphi(x, y) = \left(1-\delta_{xy}\right)\left(\logit\rho_\mathrm{same}-\logit\rho_\mathrm{different}\right),\label{eq:social-separation-sbm}
\end{equation}
where $\delta_{xy}$ is the Kronecker delta. The social separation only depends on block membership, and it is not affected by the size of each block. For members of different blocks, $\varphi$ is the difference of log odds ratios for the existence of intra-group ties as opposed to inter-group ties. The social separation is equal to the natural logarithm of the $\alpha$-index proposed by \textcite{Moody2001} for categorical attributes $x$, but $\varphi$ is applicable to arbitrary attributes and connectivity kernels.

The social separation $\varphi(x,y)$ is not sufficient to quantify segregation at the level of an individual: we also need to consider the distribution $P(y)$ of attributes $y$ of all other members of society, such as their age, sex, or other demographics. For an individual with attributes $x$, we define the \emph{social isolation}
\begin{equation}
    \phi(x) = \int dy \ P(y) \varphi(x, y),\label{eq:segregation-individual}
\end{equation}
which quantifies the average social separation between an individual with attribute $x$ and other members of society. For the SBM, we substitute \cref{eq:social-separation-sbm} into \cref{eq:segregation-individual} and obtain
\begin{equation}
    \phi(x) = (1-P(x))\left(\logit\rho_\mathrm{same}-\logit\rho_\mathrm{different}\right),\label{eq:segregation-individual-sbm}
\end{equation}
where $P(x)$ is the probability to belong to block $x$, and we have used the identity $\sum_{y=1}^K P(y)\left(1-\delta_{xy}\right)= 1 - P(x)$. Members of all blocks experience the same degree of isolation if the blocks are of the same size. If the sizes are unequal, minorities experience more isolation and majority groups experience less isolation. Indeed, ethnic minorities in schools tend to be more isolated and have fewer social ties~\cite{Currarini2009}. The off-diagonal terms of the Hessian of the social isolation $\phi$ quantify interactions, such as the joint effect of age and ethnic differences.

To understand how segregated society is as a whole, we would like to aggregate the social isolation $\phi$, but the appropriate statistic depends on the question at hand. For example, if we wanted to study the most isolated subpopulation of society, we should consider $\max_{x\in\blauspace}\phi(x)$. Here, we take a utilitarian approach and, in line with \cref{eq:segregation-individual}, define the \emph{social strain} as
\begin{equation}
    \Phi = \int dx\ P(x) \phi(x),\label{eq:social-strain}
\end{equation}
which quantifies the average social separation amongst members of the society. It is zero when individuals do not discriminate based on attributes, and it can reach arbitrarily large values in a society comprising multiple groups that are completely disconnected. For the SBM, we substitute \cref{eq:segregation-individual-sbm} into \cref{eq:social-strain} and obtain
\begin{align}
    \Phi&= \gamma\left(\logit\rho_\mathrm{same}-\logit\rho_\mathrm{different}\right),\nonumber \\
    \text{where }\gamma&=1 - \sum_{x=1}^K P^2(x) \label{eq:segregation-distributional-factor}
\end{align}
is an index of dispersion~\cite{Moody2001} and accounts for the relative sizes of the $K$ blocks. The social strain is maximal when the groups are of equal size. If one of the blocks is larger, the social strain and index of dispersion approach zero as the sizes of the minority blocks decrease: members of the majority group experiences little social isolation. It is unsurprising that there is no social strain if the society is homogeneous, but the utilitarian approach has a serious limitation: it has little concern for minorities that are not well integrated in society. For equal group sizes, the social strain increases with the number of groups, asymptotically reaching a maximum value of $\logit\rho_\mathrm{same}-\logit\rho_\mathrm{different}$.

\subsection{Distance and scale in Blau space\label{sec:segregation-distance}}

The social separation takes a simple form if the probability for two individuals to connect is a logistic kernel~\cite{Hoff2002}, i.e.\
\begin{equation}
    \logit\rho(x, y, \theta) = \sum_{l=1}^p\theta_l f_l(x,y),\label{eq:logistic-kernel}
\end{equation}
where the $p$-dimensional vector $\theta_l$ parametrises the kernel, and $f(x, y)$ is a set of $p$-dimensional features that are predictive of the connection probability, such as the age difference $f_\text{age}=\abs{x_\mathrm{age}-y_\mathrm{age}}$. Intersectionality can be accounted for by including interaction terms in the feature set. The social separation between $x$ and $y$ comprises contributions from the features of the logistic kernel:
\begin{align}
    \varphi(x,y) &= \sum_{l=1}^p\varphi_l(x, y),\label{eq:social-separation-logistic}\\
    \text{where }\varphi_l(x, y) &= \theta_l\left(f_l(y, y) - f_l(x, y)\right)\label{eq:social-separation-contribution}
\end{align}
is the contribution due to a single feature $l$. In fact, $\varphi$ is a true metric for many logistic connectivity kernels.

\begin{prop}
    The social separation $\varphi(x,y)$ is a metric if the kernel is homophilous, i.e.\ $\theta_l<0$, and each feature $f_l(x, y)$ is a constant or a positive affine transform of a metric $d_l(x, y)$, i.e.
    \begin{equation}
        f_l(x, y) = a_l d_l(x, y) + b_l,\label{eq:segregation-feature-affine}
    \end{equation}
    where $a_l > 0$ and $b_l$ are the parameters of the affine transform.
\end{prop}
\begin{proof}
    According to \cref{prop:segregations-semi-metric}, the social separation $\varphi(x,y)$ is a semi-metric, and it comprises contributions from individual features, as illustrated by \cref{eq:social-separation-logistic}. Showing that each contribution $\varphi_l(x, y)$ satisfies the triangle inequality is sufficient for $\varphi(x, y)$ to satisfy it, i.e.\ we require
    \begin{equation}
        \varphi_l(x, z) \leq \varphi_l(x, y) + \varphi_l(y, z)\label{eq:segregation-metric-inequality}
    \end{equation}
    for all $l$. Substituting \cref{eq:segregation-feature-affine} into \cref{eq:segregation-metric-inequality} yields
    \begin{equation}
        -\theta_l a_l d_l(x, z)\leq -\theta_l a_l\left[d_l(x, y) + d_l(y, z)\right],\label{eq:segregation-metric}
    \end{equation}
    where we have used the metric property $d_l(x, x) = 0$ for all $x$, and the constant $b_l$ in \cref{eq:segregation-feature-affine} vanishes by \cref{eq:social-separation-contribution}. The inequality in \cref{eq:segregation-metric} holds because $\theta_l<0$ for homophilous kernels, $a_l>0$ by assumption, and $d_l(x, y)$ is a metric. \Cref{eq:segregation-metric-inequality} is trivially satisfied for a constant feature, such as a bias term controlling the overall edge density.
\end{proof}

In other words, the social separation statistic is a true measure of \emph{distance} in the social space with a probabilistic interpretation if features are themselves measures of distance, including all the features we consider subsequently. This observation puts Peter Blau's~\cite{Blau1977} hypothesis that ``the macrostructure of societies can be defined as a multidimensional space of social positions among which people are distributed and which affect their social relations'' on a sound statistical footing: \emph{fitting conditionally-independent edge models allows us to learn the metric of Blau space}. The metric has a universal scale: \emph{one unit of social separation has the same probabilistic meaning independent of the society under consideration}, facilitating comparison across disparate datasets. Even if two societies have different Blau space dimensions, e.g.\ a society might exist that strongly discriminates based on characteristics which are not found in other societies, the social separation between a pair of individuals has a common meaning. 

\subsection{Parameter inference given ego network data\label{sec:inference}}

A representative sample of dyads between individuals together with their demographic attributes is not generally available. However, a number of surveys have collected information about the social ties of respondents using name-generator questions which elicit social ties by asking respondents to nominate their friends~\cite{Kalmijn2007}, individuals they feel close to~\cite{Hipp2009}, or discussion partners~\cite{Marsden1987,McPherson2006}. To generate examples of disconnected dyads, we consider a random sample of pairs of individuals. To account for this non-ignorable data collection process, we introduce a variable $I_{ij}\in\{0,1\}$ indicating whether a particular dyad $A_{ij}$ was observed~\cite[chapter 8]{Gelman2013}. The available data thus comprise the demographic attributes $x$ of individuals included in the sample and the dyad state $A_{ij}$ (1 if $i$ and $j$ are connected and $0$ otherwise) if it was observed, i.e.\ $I_{ij}=1$. Adapting the argument presented by \textcite{King2001} to a Bayesian paradigm, we consider the posterior distribution over kernel parameters $\theta$ given the available data:
\begin{equation}
    P(\theta|A, f, I=1)\propto P(A|\theta,f,I=1)P(\theta),\label{eq:parameter-posterior}
\end{equation}
where $P(\theta)$ is the kernel parameter prior, and $f=f(x,y)$ are features sufficient to evaluate the connectivity kernel given demographic attributes $x$ and $y$. The observed-data likelihood is
\begin{equation}
    P(A|\theta,f,I=1)=\frac{P(f|A,\theta,I=1)P(A|\theta,I=1)}{P(f|\theta,I=1)}.\label{eq:observed-likelihood}
\end{equation}
Considering the first term in the numerator of \cref{eq:observed-likelihood}, we note that the distribution over kernel features given the state $A$ of the dyad does not depend on whether it was included in the sample or not. More formally,
\begin{align}
    P(f|A,\theta,I=1)&=P(f|A,\theta)\nonumber\\
    &=\frac{P(A|f,\theta)P(f|\theta)}{P(A|\theta)}.\label{eq:conditional-equivalence}
\end{align}
Turning to the denominator in \cref{eq:observed-likelihood}, we find
\begin{align}
    P(f|\theta,I=1)&=\sum_{\alpha=0}^1 P(f|A=\alpha,\theta,I=1)P(A=\alpha|\theta,I=1)\nonumber\\
    &=P(f|\theta)\sum_{\alpha=0}^1 P(A=\alpha|f,\theta)\frac{P(A=\alpha|\theta,I=1)}{P(A=\alpha|\theta)},\label{eq:likelihood-denominator}
\end{align}
where we used the identity in \cref{eq:conditional-equivalence} to arrive at the second line. Substituting \cref{eq:conditional-equivalence,eq:likelihood-denominator} into \cref{eq:observed-likelihood}, the observed-data likelihood is
\begin{align}
    P(A|\theta,f,I=1)&=\frac{P(A|f,\theta)r(A)}{\sum_{\alpha=0}^1 P(A=\alpha|f,\theta)r(\alpha)},\label{eq:observed-likelihood-final}\\
    \text{where }r(\alpha)&=\frac{P(A=\alpha|\theta,I=1)}{P(A=\alpha|\theta)}\nonumber
\end{align}
is the ratio of prevalences of dyad state $\alpha$ in the sample and the general population. In practice, we approximate the prevalence ratio $r$ using the empirical sample prevalence and prior knowledge about the prevalence in the population. The posterior can be evaluated by substituting \cref{eq:observed-likelihood-final} into \cref{eq:parameter-posterior}, and we can thus infer the parameters $\theta$ from ego network data. See \cref{app:numeric-stability} for details on how to evaluate the observed-data log-likelihood in a numerically stable fashion and \cref{app:inference-validation} for a validation of the inference methodology using synthetic data. For logistic connectivity kernels, the observed-data likelihood in \cref{eq:observed-likelihood-final} resembles a conventional case-control likelihood, e.g.\ as used by \textcite{Smith2014}.

\section{Application\label{sec:application}}

\subsection{Ego network data collected in surveys}

The social ties identified through name-generator questions depend on the nature of the relationship, the mode of administration of the questionnaire (e.g.\ face-to-face, telephone interview, or online survey), and the interviewer~\cite{Marin2004,Eagle2015}. Consequently, we do not expect the kernel parameters inferred from different datasets to be completely consistent. In the following investigation of ego networks, we restrict the nature of relationships to friends who are not relatives as much as the available data permit: we are interested in \emph{voluntary} association amongst members of the population rather than the social structures they were born into~\cite{Kalmijn2007}.

Demographic information about nominees can be collected either by asking seeds about their friends' demographic background~\cite{Marsden1987,McPherson2006} or by conducting follow-up surveys with nominated friends~\cite{Johnson1989}. The latter seems preferable because respondents may not have complete information about their social contacts. For example, the age of nominees in the British Household Panel Survey (BHPS), a dataset we consider in \cref{sec:survey-bhps-usoc}, is 60\% more likely to be an integer multiple of ten than it is for seeds---presumably because seeds round the age of their friends to the nearest decade. In anticipation of such challenges, the coding for the nominees is often coarser than for seeds. To compare the demographic attributes of seeds and nominees we need to unify the coding (see \cref{app:coding-and-features} for details for each dataset). Unfortunately, follow-up surveys require additional resources to interview the nominees and may suffer from low response rates.

\subsection{General Social Survey\label{sec:survey-gss}}

The General Social Survey (GSS) is a nationally-representative face-to-face survey of non-institutionalised adults living in the US. Demographic attributes of seeds are collected regularly and include age, sex, ethnicity, religion, and education~\cite{Marsden1987,Marsden1988}. In 2004, respondents were asked about the demographic background of people ``with whom they discuss important matters'', which tends to elicit close ties~\cite{Marin2004}. We omit all nominees who are not considered to be friends or who are family. Some of the demographic attributes of seeds and nominees are missing because respondents did not know or refused to provide the information, and we drop dyads associated with individuals with one or more missing attributes, as shown in \cref{tbl:survey-sample-size}. Such a complete-case analysis can introduce biases if the data are not missing completely at random, but handling the missing data in a principled fashion would require us to develop a model for demographic attributes~\cite{Pigott2001}.

The coding of age and sex is consistent amongst seeds and nominees. We aggregate the detailed coding of ethnic and religious attributes of seeds to match the coding of nominees, as shown in \cref{tbl:survey-gss-coding}. Kernel features include the absolute age and ordinal education level difference as well as binary indicators for differences along the sex, ethnicity, and religion dimensions. For each demographic attribute, we define a feature for the logistic kernel in \cref{eq:logistic-kernel}, as shown in \cref{tbl:survey-gss-coding}. To standardise the features $f(x_i,x_j)$, we subtract their mean and divide non-binary features by twice their standard deviation~\cite{Gelman2008a}; binary features are not rescaled. The statistics are calculated with respect to a random sample of pairs of seeds. Feature standardisation allows us to compare kernel parameters more easily~\cite{Gelman2008a} and simplifies the formulation of priors: we use independent, weakly-informative Cauchy priors for the kernel parameters such that
\begin{equation*}
    P(\theta_l) \propto {\left[1+{\left(\frac{\theta_l}{\alpha_l}\right)}^2\right]}^{-1}.
\end{equation*}
Following \textcite{Gelman2008}, we chose the scale parameters $\alpha_l=2.5$ for $l>1$ to represent our weak prior belief that changing a feature by one standard deviation is unlikely to change the log odds by more than five: the independent Cauchy distributions regularise the kernel parameters by placing significant prior probability near zero, but their heavy tails allow for significant departures from zero should the data be in support of large parameters. We set $\alpha_1=10$ because the parameter $\theta_1$ associated with the constant bias term could change significantly depending on the population size~\cite[chapter 16]{Gelman2013}.

\begin{figure}
    \includegraphics{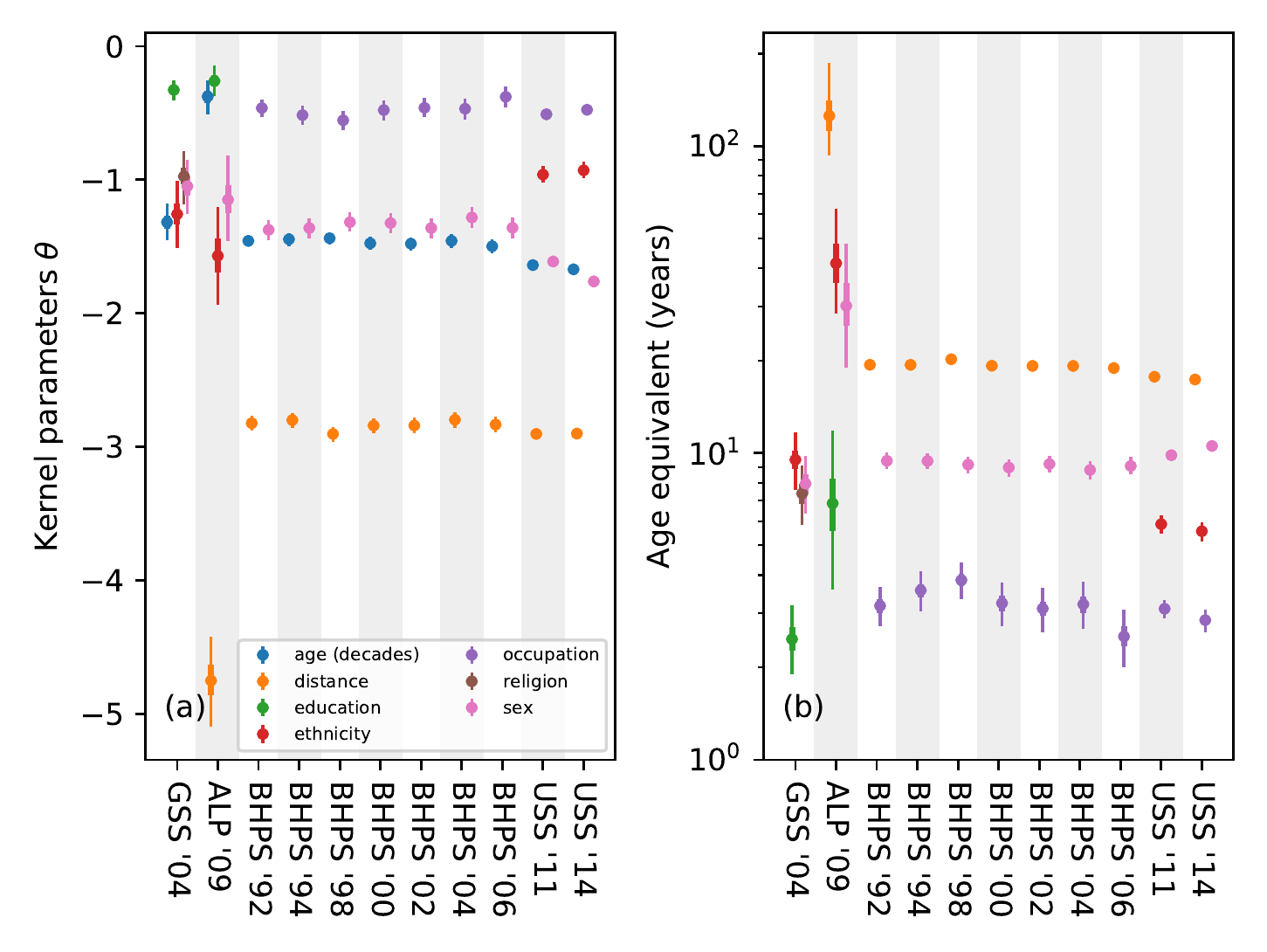}
    \titlecaption{Age and physical separation have a strong impact on connection probabilities, and converting parameters into age equivalents makes feature comparison more intuitive.}{
    Panel~(a) shows kernel parameters inferred from ego network data for each dataset.
    Panel~(b) shows age equivalents. For binary features (sex, occupation, religion, ethnicity, and distance for the American Life Panel), the equivalent number of years corresponds to a change from having the same attribute to having a different attribute. Age equivalents for the American Life Panel are overestimated (see \cref{sec:survey-alp} for details).
    Markers represent the posterior median, thick error bars the interquartile range, and thin error bars the 95\% credible interval.\label{fig:parameters-and-age-equivalents}}
\end{figure}

The inference is performed in two steps: first, we maximise the posterior with respect to the parameters $\theta$ using a gradient ascent algorithm. Second, we run a Metropolis-Hastings algorithm to draw samples from the posterior~\cite{Hastings1970}. Summary statistics of the posterior are shown in \cref{fig:parameters-and-age-equivalents}~(a). The connection probabilities decrease quickly with increasing age differences: the odds of connection are reduced by a multiplicative factor of about 0.3 per decade. Ethnic, sex, and religious differences all seem to have a similar effect and decrease the odds by a factor of about 0.3 each; a difference of one educational level reduces the odds by a factor of 0.7.

\textcite{Hipp2009} used the logarithm of physical separation as a benchmark to translate the effect of other attributes into distance-equivalents. We instead use age as a benchmark because age is available for most datasets and is typically coded uniformly in years. In contrast, physical separation is often not available or coded heterogeneously across different datasets. For example, the American Life Panel only provides location information at the state level (see \cref{sec:survey-alp}), whilst the British Household Panel Survey recorded distance between seeds and nominees as ordinal data (see \cref{sec:survey-bhps-usoc}). For the GSS, being of a different ethnicity is equivalent to a nine-year age difference, and having a different sex or religion translates to eight and seven years, respectively. One educational level, as defined in \cref{tbl:survey-gss-coding}, corresponds to three years, as shown in \cref{fig:parameters-and-age-equivalents}~(b).

\subsection{American Life Panel\label{sec:survey-alp}}

The American Life Panel (ALP) is a nationally-representative panel of adults resident in the US~\cite{Pollard2017}.\@ Panel members are interviewed either using their own internet connection or are provided with a web television to access surveys. Data are collected regularly and each survey has a different focus. In 2009, information about social networks and financial literacy was collected. Demographic attributes included sex, age, ethnicity, education, their state of residence, and whether respondents identified as Hispanic. Respondents were also asked to nominate others with whom they ``discuss financial matters''~\cite{Mihaly2009}. We only include nominees who are friends of seeds and exclude kinship ties; see \cref{tbl:survey-alp-coding} for details of harmonisation of attributes across seeds and nominees.

Homophily with respect to sex and ethnicity is slightly stronger than in the GSS, and educational homophily is weaker, but the inferred parameters are broadly consistent with the GSS.\@ Age differences appear to play less of a role in the discussion of financial matters at first sight, but the inference is severely biased for age. We cannot resolve strong age homophily because data are only recorded in 15-year bins: the small age parameter is likely a result of regression dilution caused by measuring ages imprecisely~\cite{Hutcheon2010}. Consequently, the age equivalents in \cref{fig:parameters-and-age-equivalents}~(b) are inflated. Being resident in a different state has by far the most significant impact on friendship formation.

\subsection{British Household Panel Survey and Understanding Society\label{sec:survey-bhps-usoc}}

The British Household Panel Survey (BHPS) was a nationally-representative face-to-face survey in the UK.\@ It was conducted from 1991 to 2008 and has since been replaced by the Understanding Society Survey (USS). Respondents were asked questions about ``their closest friends'' every other year as part of the BHPS and every three years in the USS. Data include sex, age, occupational status, ethnicity (only in the USS), and how far away friends live~\cite{Institute-for-Social-and-Economic-Research2000,Institute-for-Social-and-Economic-Research2017} (see \cref{tbl:survey-usoc-coding} for details).

The inferred kernel parameters are largely consistent with the inference for the ALP and GSS in the US suggesting that friendship formation proceeds similarly in the two countries. We have omitted data from the BHPS in 2008 because we identified errors in the coding which have since been confirmed by the Institute of Social and Economic Research~\cite{Hoffmann2016}. Similarly, we omitted data from the BHPS in 1996 because physical separation between friends was not recorded. As shown in \cref{fig:parameters-and-age-equivalents}, homophily seems to have increased in recent years, but the changes are likely the result of a change in methodology rather than a change in behaviour: the BHPS collected friendship information as part of the main survey, whereas the USS used a self-completion questionnaire~\cite{Hoffmann2017}.

\subsection{Inferred segregation}

\begin{figure}
    \includegraphics{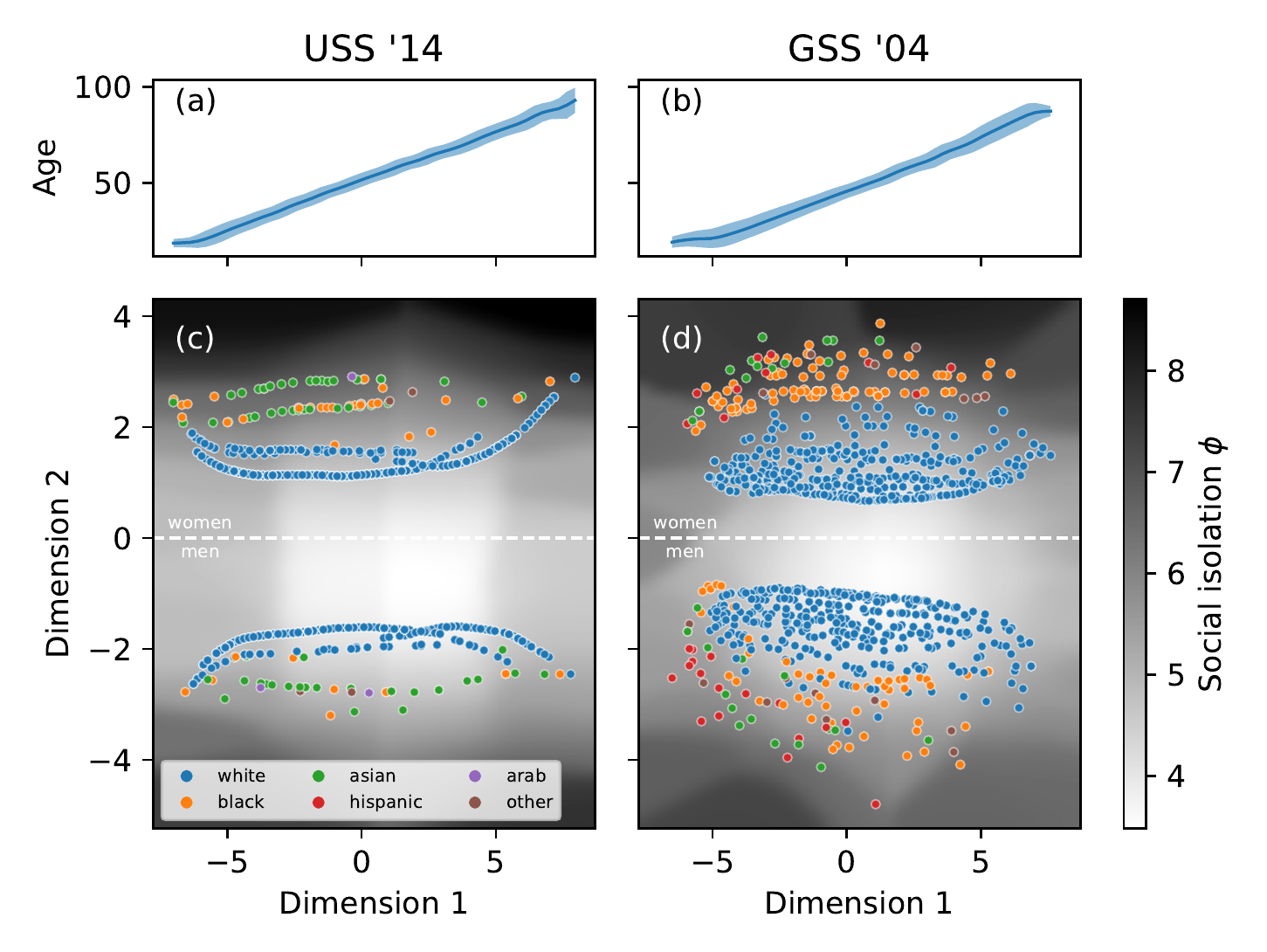}
    \titlecaption{A lower-dimensional embedding of the inter-node distances reveals an interpretable social space in the UK and US.}{Panels~(a) and~(b) show the mean and standard deviation of ages as a function of the first embedding dimension as a solid line and a shaded region, respectively. Panels~(c) and~(d) show a scatter plot of respondents in a two-dimensional embedding space whose coordinates were obtained from the social separation $\varphi$ using multidimensional scaling. The colour of a marker indicates the respondent's ethnicity. The heat map represents a smoothed estimate of the social isolation $\phi$. The ``bands'' of individuals in panel~(c) correspond to different occupational statuses, such as employed or retired.\label{fig:segregation-embedding}}
\end{figure}

To get a better understanding of Blau space and the metric induced by the connectivity kernel, we consider a sample $S$ of 1,000 respondents from the GSS and USS. For each sample, we compute the social separation between pairs of respondents to obtain a distance matrix
\begin{equation*}
    \hat\varphi_{ij} = \transpose{\hat\theta}\left(f(x_j, x_j)-f(x_i, x_j)\right),
\end{equation*}
where $\hat\theta$ is the posterior median of the kernel parameters discussed in \cref{sec:survey-gss,sec:survey-bhps-usoc}. We use multidimensional scaling to embed the respondents in a two-dimensional space~\cite{Borg1996}, as shown in \cref{fig:segregation-embedding}. Panels~(c) and~(d) show the two-dimensional embedding that best approximates the distance matrix in the high-dimensional social space (we omit contribution due to physical space for the USS to make the embeddings comparable).

The first dimension captures the age of respondents, as illustrated in panels~(a) and~(b): the mean age increases monotonically as a function of the first embedding dimension, and the standard deviation is small. We evaluated both statistics using Gaussian kernel smoothing~\cite[chapter~6]{Hastie2009}. The second dimension captures sex and ethnicity as well as occupational status (for the USS) and education and religion (for the GSS). As expected from \cref{eq:segregation-individual-sbm}, ethnic minorities are more isolated and live on the outskirts of society while the ethnic majority occupies the centre. The embedding suggests that age has the strongest impact on how people form friendships.

Panel~(c) and~(d) of \cref{fig:segregation-embedding} also show the social isolation $\phi$ experienced by individuals as a greyscale heat map which we obtained in two steps: first, we evaluated an estimate of the social isolation
\begin{equation*}
    \hat\phi_i = \frac{1}{\card{S}-1}\sum_{j\in S:j\neq i} \hat\varphi_{ij}.
\end{equation*}
Second, we applied Gaussian kernel smoothing to the social isolation in the embedding space. Respondents occupying the centre of society experience little isolation whereas individuals in the periphery are more isolated. For example, members of the ethnic majority experience an average social isolation of 4.53 (4.47--4.59 95\% credible interval) in the USS and 4.46 (4.10--4.84 95\% credible interval) in the GSS. In contrast, the average social isolation amongst ethnic minorities is 4.99 (4.91--5.07 95\% credible interval) and 5.16 (4.73--5.58 95\% credible interval): significantly higher than for the ethnic majority.

\begin{figure}
    \includegraphics{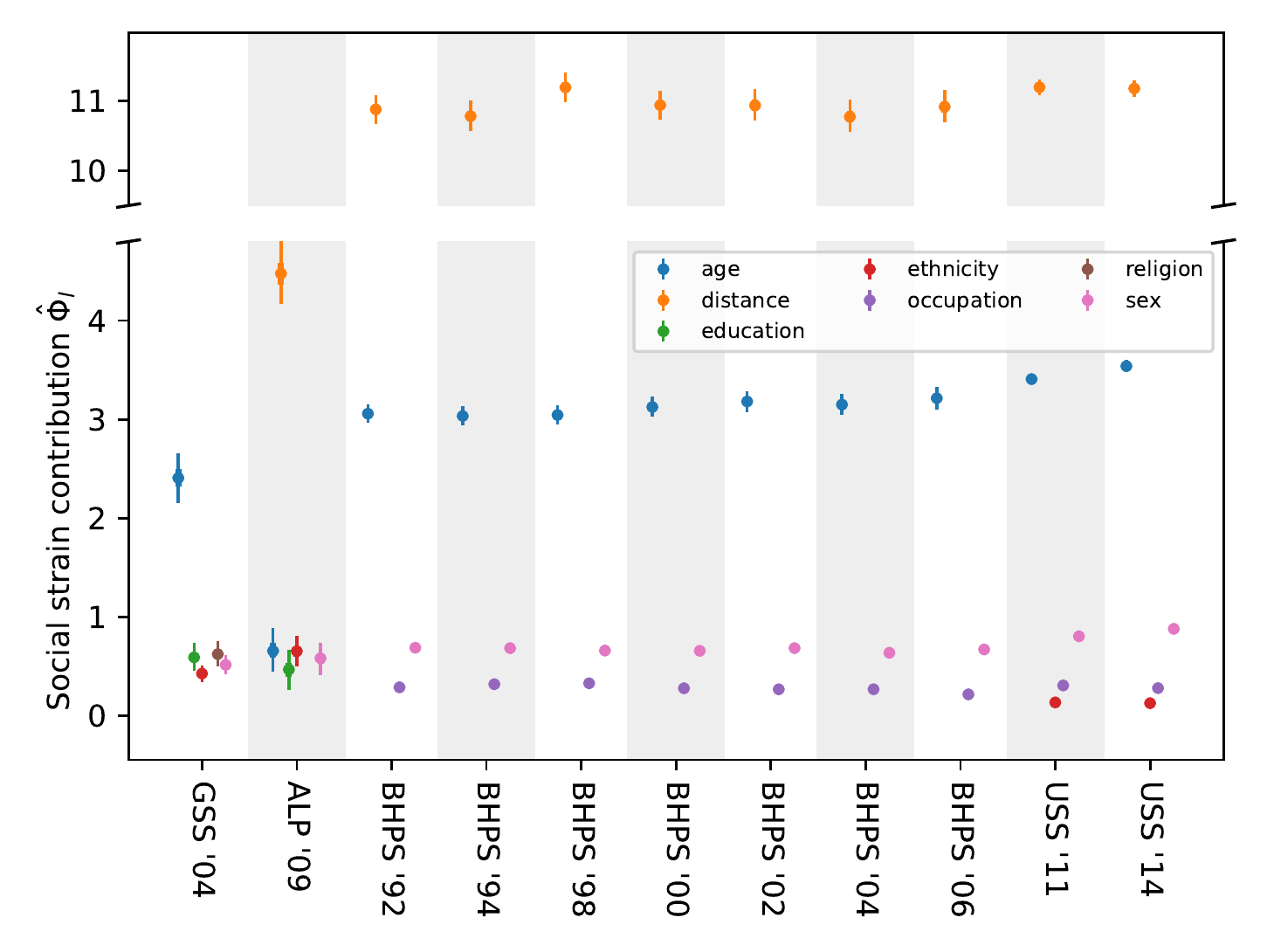}
    \titlecaption{Physical separation and age differences are the most important factors preventing integration of society for all datasets.}{Markers represent the posterior median of the contributions to social strain for each feature, thick error bars the interquartile range, and thin error bars the 95\% credible interval.\label{fig:segregation-survey}}
\end{figure}

Similar to the social separation in \cref{eq:social-separation-contribution}, the social strain can be broken down into components
\begin{equation}
    \Phi_l = \theta_l \int dx\ dy\ \left(f_l(y, y) - f_l(x, y)\right) P(x) P(y)\label{eq:segregation-contributions}
\end{equation}
because it is a linear functional of $\varphi$: each component contributes to the social strain in society. For each of the datasets, we evaluate an estimate of the contributions to the social strain
\begin{equation*}
    \hat\Phi_l = \frac{2\theta_l}{\card{\seeds}(\card{\seeds}-1)} \sum_{i < j\in\seeds:i\neq j} \left[f_l(x_j, x_j)-f_l(x_i, x_j)\right],
\end{equation*}
where the sum is over all distinct pairs of seeds $\seeds$. The contribution $\Phi_l$ quantifies the average social separation due to feature $l$, and it captures the effect of both the connectivity kernel $\rho(x,y)$ and the attribute distribution $P(x)$ in Blau space: neither is sufficient on its own to quantify segregation. Furthermore, $\Phi$ and its contributions in \cref{eq:segregation-contributions} can facilitate comparison across different datasets, as illustrated in \cref{fig:segregation-survey}: they have an intuitive interpretation (the probability of edges decreases with increasing segregation) and universal scale (one unit of segregation has the same effect on edge probability).

As might be expected based on previous studies~\cite{Lambiotte2008, Expert2011, Backstrom2010, Scellato2011, Illenberger2013}, physical space has by far the most significant impact on how people connect with one another. Age homophily places the second strongest restriction social connections, and it is more than three times as restrictive as any other feature except physical space. The ALP survey is an exception because of the regression dilution~\cite{Hutcheon2010} discussed in \cref{sec:survey-alp}. Age homophily is known to be particularly strong for friendship networks~\cite{McPherson2001}. Homophily with respect to sex, ethnicity, education, religion, and occupation make similar, smaller contributions to the segregation of friendships. Importantly, the social strain captures the average contribution to social isolation: it can be small either because there is little homophily or because there is a large majority group, as evident from~\cref{eq:segregation-distributional-factor}. For example, almost 80\% of respondents in the GSS identify as ``white'' and experience little social isolation due to their ethnicity, whereas minority groups experience more social isolation. On average, social isolation due to ethnicity is small.

\section{Discussion\label{sec:discussion}}

We considered a generative model for social networks embedded in Blau space, a space spanned by the demographic attributes of members of society. We developed a family of segregation statistics with a universal scale (since they are based on the common notion of the probability of a social tie), facilitating comparison between datasets collected at different times or in different cultural contexts. Furthermore, the segregation statistics are applicable to mixed attribute types, have a natural reference point, and an intuitive interpretation: the probability to form connections decreases with increasing segregation. They are applicable at different resolutions: connections, individuals, and society as a whole. For certain logistic connectivity kernels, the social separation is a metric for Blau space and allows us to quantify social distance in a principled fashion. The model-based approach facilitates the study of segregation in synthetic social networks, the effect of interventions, and principled quantification of uncertainties in an applied setting.

Based on eleven ego network datasets collected in the United Kingdom and United States, we inferred the connectivity kernel $\rho(x,y)$, i.e.\ the probability for an individual with demographic attributes $x$ to connect with another with attributes $y$. Using the kernel, we compared segregation across different datasets along different demographic dimensions and found that physical distance and age have the most significant impact on how well society is connected. We used the Blau space metric to evaluate the social distance amongst respondents of the GSS and USS. Using a lower-dimensional embedding of the respondents, we explored Blau space, corroborating our findings that age has a profound impact on restricting friendship formation.

The importance of physical distance highlights that our suite of segregation statistics does not distinguish between \emph{choice homophily} and \emph{opportunistic homophily}~\cite{Franz2010}. The former is a result of individuals having an active preference to connect with others who are alike, whereas the latter is a result of individuals being exposed to others who are similar to them. Opportunistic homophily is likely to be a large contributing factor to spatial homophily because individuals are less likely to encounter people who live far from them. Similarly, the statistics do not discriminate between choice homophily and \emph{social influence}, i.e.\ the tendency for people to become more alike given a connection~\cite{Shalizi2011}.

Other features, including sex, ethnicity, religion, education, and occupation, have smaller effects on the presence of connections. Notably, the social strain due to ethnicity in the GSS and ALP is larger than in the USS: first, the effect of ethnicity is more pronounced in the US, as shown in \cref{fig:parameters-and-age-equivalents}. Second, society in the US is more ethnically diverse than in the UK (79\% white in GSS '04 compared with 89\% white in USS '14), increasing strain on average, as exemplified with a SBM in \cref{sec:model-based-segregation}.

Even though we did not expect the kernel parameters to be consistent across countries, time, or even different surveys, people connected with one another in a surprisingly similar fashion across the different datasets (the BHPS and USS are longitudinal studies such that consistent parameter estimates are less surprising). Our observations, together with a study by \textcite{Mossong2008} finding that ``mixing patterns [\ldots] were remarkably similar across different European countries'', suggest that connectivity kernels for friendships vary little across societies and time. To test this hypothesis, further surveys should be conducted in a unified fashion to minimise the effects of question wording and how the survey is administered~\cite{Eagle2015}. In particular, such surveys should explore options to explicitly incentivise nominees to provide data about themselves~\cite{Biernacki1981}: seeds may not recall certain attributes, or nominees may deliberately portray themselves inaccurately~\cite{Bruch2016}. Questions regarding ethnicity should allow respondents to provide multiple answers such that people with mixed ethnic backgrounds can express their identity. Rather than asking respondents about potentially sensitive information, such as income, proxy information that is more readily available---and potentially more informative of how individuals interact with society---could be collected~\cite{Po2012}. Whenever possible, aggregation of attributes such as age into bins should be avoided because it limits the ability to infer kernel parameters~\cite{Hutcheon2010}, as we saw in \cref{sec:survey-alp}. Connectivity kernels should be inferred jointly for all dimensions of Blau space to control for social preferences on correlated attributes.

The connectivity kernel is an intuitive model of how people connect with one another, and it is able to reproduce some of the statistics of real social networks. For example, people in high-density regions of Blau space have been observed to have more connections~\cite{Currarini2009}. However, exponential random graph models may be able to better capture the nature of social networks~\cite{Wimmer2010}. Furthermore, we have used a connectivity kernel that (a)~is symmetric and cannot identify whether there is a status order in society~\cite{Chan2004,Ball2013} and (b)~only depends on differences between individuals. For example, young men tend to have more social contacts than young women, and older women have more social contacts than older men~\cite{Bhattacharya2016}---an observation that cannot be captured by a kernel of the form we have considered. The connectivity kernel could be refined by adding the demographic attributes of the seeds and nominees as features, capturing sociability and popularity, respectively. Furthermore, it should be determined whether the number of ``intervening opportunities''~\cite{Stouffer1940}, absolute distance in Blau space, or a hybrid thereof are most predictive of tie probability. Ultimately, learning a connectivity kernel without a pre-specified parametric form should be considered~\cite{Frolich2006} because they can better capture complex patterns, such as interactions between different demographic attributes. We note that, irrespective of the choice of connectivity kernel, the interpretable segregation statistics considered here remain valid and useful.

\section*{Acknowledgements}

We would like to thank Sahil Loomba and two anonymous referees for useful feedback on the manuscript and acknowledge support from the grant EP/N014529/1.

\printbibliography
\end{refsection}

\appendix
\numberwithin{equation}{section}
\numberwithin{figure}{section}
\numberwithin{table}{section}

\begin{refsection}

The code to reproduce the results and figures is available at \url{https://github.com/tillahoffmann/kernels}.

\section{Evaluation of the observed-data log-likelihood}

\subsection{Weighting to account for non-uniform inclusion probabilities}

Seeds are often not included in the survey uniformly at random, and weights are traditionally used to compensate for the potentially biased selection of respondents~\cite{Kish1992}. Including weights in Bayesian analyses is generally difficult~\cite{Gelman2007}, and, in principle, we should model the data collection process explicitly~\cite[chapter~8]{Gelman2013}. Unfortunately, modelling the data collection process is non-trivial, and we use a weighted pseudo-likelihood instead~\cite{Pfeffermann1996}. In particular, the observed-data log-likelihood from \cref{eq:observed-likelihood-final} becomes
\begin{multline}
    L = \sum_{(i,j):I_{ij} = 1} w_j\left\{A_{ij} + (1 - A_{ij})w_i\right\}
    \\\times\left\{A_{ij}\log\rho_{ij} + (1-A_{ij})\log(1-\rho_{ij})-\log\left[r(0)(1 - \rho_{ij}) + r(1)\rho_{ij}\right]\right\},\label{eq:survey-weighted-log-likelihood}
\end{multline}
where $w_j$ is the weight associated with seed $j$. We clip all weights exceeding the \ordinalnum{95} percentile of the empirical weight distribution and normalise them such that $\sum_{j\in\seeds}w_j=\card{\seeds}$. Censoring the weights, also known as Winsorisation, limits the variance induced by attributing variable importance to different observations at the expense of introducing a small bias~\cite{Kish1992}.

\subsection{Numerical stability\label{app:numeric-stability}}

The evaluation of the observed-data log-likelihood may suffer from numerical instabilities, especially when the connectivity kernel $\rho(x, y, \theta)$ is small. We can mitigate such instabilities for logistic connectiviy kernels, i.e.
\begin{align}
    \rho(x,y,\theta)&=\expit(\transpose{\theta}f(x,y)),\\
    \text{where }\expit(\xi)&=\frac{1}{1+\exp(-\xi)}
\end{align}
is the logistic function. In particular, note that $1 - \expit(\xi) = \expit(-\xi)$ and $\log\expit(\xi)=-\logp\left[\exp(-\xi)\right]$, where $\logp(\xi)=\log(1 + \xi)$ is a numerically stable implementation---even for $\abs{\xi}\ll 1$. Substituting into \cref{eq:observed-likelihood-final} yields
\begin{multline}
    \log P(A|f,\theta,I=1)= -\sum_{(i,j):I_{ij}=1} A_{ij}\logp\exp\left(-\transpose{\theta}f_{ij}\right)+\left(1-A_{ij}\right)\logp\exp\left(\transpose{\theta}f_{ij}\right) \\
    + \logsumexp\left[\log r(0) - \logp\exp\left(\transpose{\theta}f_{ij}\right) + ,\log r(1)-\logp\exp\left(-\transpose{\theta}f_{ij}\right)\right],
\end{multline}
where $\logsumexp(x_1, \ldots, x_k)=\log\sum_{i=1}^k\exp(x_k)$ is a numerically stable implementation.

\section{Validation of inference methodology using synthetic ego network data\label{app:inference-validation}}

\begin{figure}
    \includegraphics{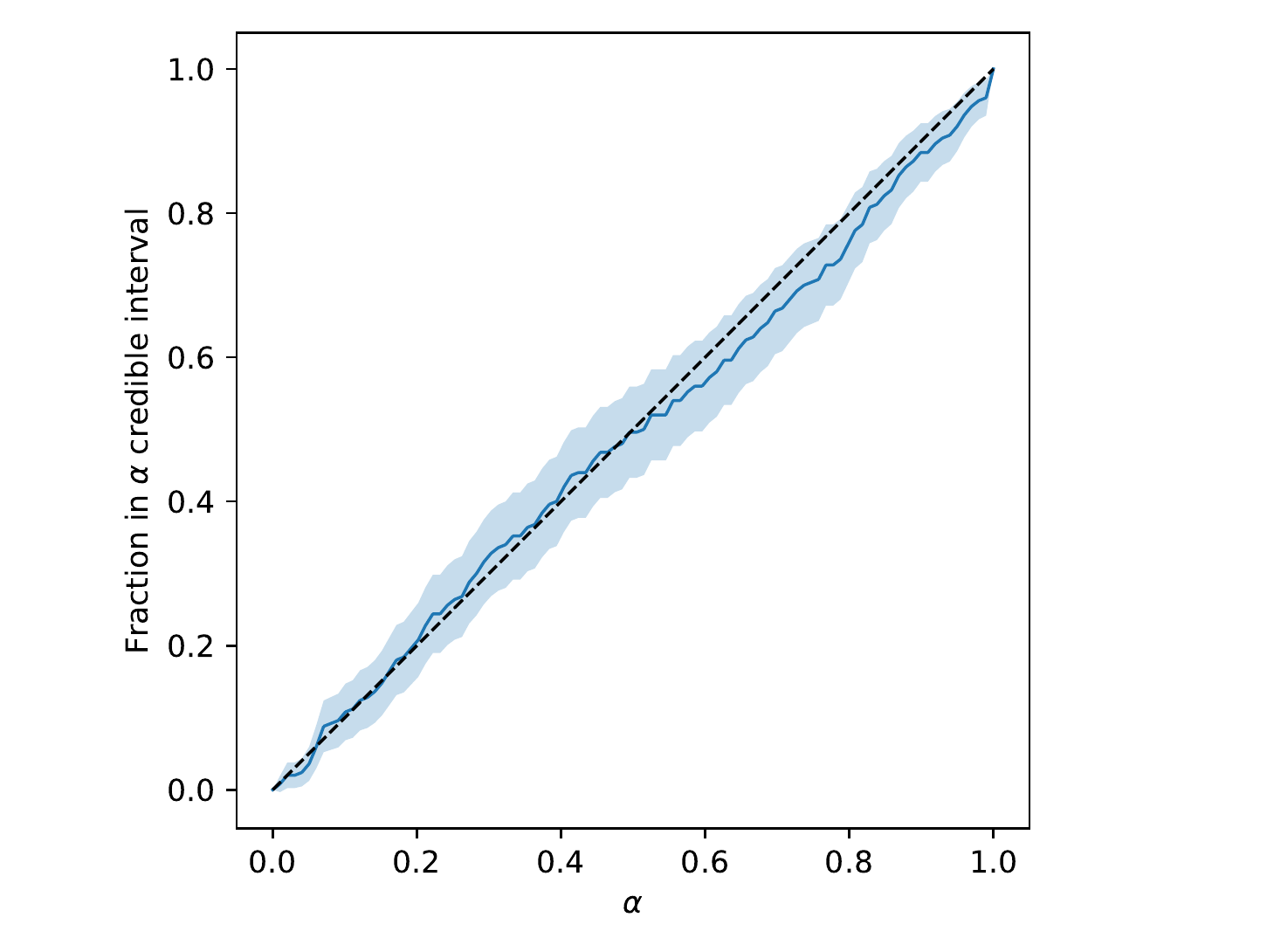}
    \titlecaption{A coverage analysis of posterior credible intervals validates the inference methodology.}{The blue line shows the fraction of inferences for which the true parameter values are contained in the $\alpha$-credible interval of a Laplace approximation of the posterior. The shaded region corresponds to two standard deviations of the mean across 250 simulations.\label{fig:credible-coverage}}
\end{figure}

To test the inference methodology, we conduct a coverage analysis of posterior credible intervals in three steps: fist, we generate 250 synthetic ego network datasets with known kernel parameter values. Second, we infer the parameter posterior distribution. Third, we evaluate the proportion $\lambda(\alpha)$ of true parameter values contained in the $\alpha$-credible interval across multiple synthetic datasets. We expect the true parameter values to be contained in the $\alpha$-credible interval for a proportion $\alpha$ of the synthetic datasets~\cite[section~10.7]{Gelman2013}, i.e.\ $\lambda(\alpha)\approx\alpha$.

In the first step, we draw the positions $x$ of $n=2,000$ nodes uniformly at random from the unit square, and we connect nodes to one another according to a logistic connectivity kernel with features
\begin{equation*}
    f(x_i, x_j) = \left(1, \frac{3 \left[\abs{x_{i1}-x_{j1}} - \frac{1}{3}\right]}{\sqrt{2}}, \frac{3 \left[\abs{x_{i2}-x_{j2}} - \frac{1}{3}\right]}{\sqrt{2}}\right),
\end{equation*}
where the first feature represents the bias, and the last two features capture distance in Blau space. The features were chosen to be standardised in the same fashion as described in \cref{sec:survey-gss}, i.e.\ to have zero mean and a standard deviation of $0.5$~\cite{Gelman2008a}. The corresponding parameter values $\theta$ are drawn from a normal distribution with unit variance and expectation
\begin{equation*}
    \left\langle \theta\right\rangle = \left(-7, 0, 0\right).
\end{equation*}
The expectation $\langle\theta\rangle$ was chosen such that it gives rise to a typical degree of $2,000 \times \expit(-7)\approx 1.8$, similar to the real-world datasets considered in \cref{sec:application}. We select $s=100$ respondents as egos and include all their alters as positive examples. We sample three times as many negative examples by selecting distinct pairs of respondents. If the number of respondents is not sufficient to draw the desired number of distinct control pairs, we use all $\frac{s(s-1)}{2}$ possible distinct pairs of respondents as negative examples.

In the second step, we maximise the log-posterior using a gradient ascent algorithm to obtain the MAP estimate $\theta$ and consider the Laplace approximation of the posterior~\cite[section~4.4]{Bishop2007}, i.e.\ a multivariate normal approximation in the vicinity of the MAP estimate. We evaluate the Hessian $H$ of the negative log-posterior at the MAP estimate to obtain the precision matrix of the Laplace approximation. We do not draw samples from the posterior because the Laplace approximation is computationally more convenient.

In the last step, we consider the quantity
\begin{equation}
   \chi^2 = \transpose{\left(\theta-\hat\theta\right)}H\left(\theta-\hat\theta\right)\label{eq:survey-chi2-statistic}
\end{equation}
which we expect to follow a $\chi^2$ distribution with three degrees of freedom~\cite{Slotani1964}. We calculate the statistic in \cref{eq:survey-chi2-statistic} for each simulation and consider the empirical probability that $\chi^2$ does not exceed the expected quantiles of the $\chi^2$-distribution, as shown in \cref{fig:credible-coverage}. As expected, the $\alpha$-credible interval contains the true parameter values for a fraction $\alpha$ of inferences, validating the inference methodology.

\section{Coding of demographic attributes and feature maps\label{app:coding-and-features}}

In the BHPS and USS, distance was coded as an ordinal variable: less than one mile, less than five miles, less than fifty miles, and more than fifty miles. We rely on having complete information about seeds to evaluate the control features in \cref{eq:logistic-kernel}. But data on the residential location of seeds is not made available to protect their privacy. Fortunately, we can sample the home locations of respondents\footnote{Sampling home locations cannot reproduce any correlation between home location and other demographic attributes.} using population estimates and the geographic boundaries of lower layer super output areas (LSOAs). LSOAs are census reporting areas and have a few thousand inhabitants each~\cite{Lsoas}. We approximate the distribution of distances between residents of the UK using rejection sampling: first, choose a LSOA with probability proportional to the number of residents. Second, choose one of the polygons associated with the LSOA with probability proportional to the area of the polygon (LSOAs are not necessarily contiguous). Third, sample points uniformly inside the bounding box of the polygon until a point inside the polygon is sampled. The last two steps assume uniform population densities within each LSOA, which is unlikely to be problematic as they are small areas. Having sampled the residential location of two respondents, we calculate the distance between respondents and cast to the same ordinal scale as reported for nominees. The USS furthermore distinguishes between friends living more than fifty miles apart but within the UK and friends outside the UK. We discard the latter (2.6\% and 2.1\% of all friends in waves C and F of the USS) because it is difficult to define an appropriate control population. For the BHPS, we implicitly assume that all friends are resident in the UK.

In the USS, respondents could identify with mixed ethnicities, and we coded such responses as a mixed membership. For example, a respondent who indicated ``mixed Asian and White'' would belong to both White and Asian ethnicities. To quantify how different two people are in terms of ethnicity, we define the feature map
\begin{equation*}
    f_\text{ethnicity}(x_i, x_j) = \frac{1}{2} \sum_{l\in E}\abs{x_{il}-x_{jl}},
\end{equation*}
where $E$ is the set of attributes encoding ethnic identity, and ethnicity memberships are normalised such that $\sum_{l\in E}x_{jl}=1$ for all $j$. For example, $f_\text{ethnicity}(x_i, x_j)=1$ for two people $i$ and $j$ one of which identifies as white and the other as black. For a person $i$ identifying as black and another person $j$ identifying as mixed black and Asian, $f_\text{ethnicity}(x_i, x_j)=0.5$.

\begin{table}
    \begin{tabularx}{\columnwidth}{lXXc}
        \toprule 
        Variable & Seed coding & Nominee coding & $f(x, y)$ \\
        \midrule 
        Bias term &\twocol{\dotfill}& $1$\\
        Age & \twocol{Age in years\dotfill} & $\left|x-y\right|$\\
        Sex & \twocol{(a)~Male, (b)~Female\dotfill} & $x\neq y$\\
        Ethnicity & (a)~\{Asian Indian, Chinese, Filipino, Japanese, Korean, Vietnamese, Other Asian\} (b)~Black, (c)~Hispanic, (d)~White, (e)~\{American Indian or Alaska Native, Native Hawaiian, Guamanian or Chamorro, Samoan, Other Pacific Islander, Other\}& (a)~Asian, (b)~Black, (c)~Hispanic, (d)~White, (e)~Other & $x \neq y$\\
        Religion & (a)~Protestant, (b)~Catholic, (c)~Jewish, (d)~None, (e)~\{Other, Buddhism, Hinduism, Islam, Orthodox, Christian, Native American, Nondenominational\} & (a)~Protestant, (b)~Catholic, (c)~Jewish, (d)~None, (e)~Other & $x\neq y$ \\
        Education & (1)~1--6 years, (2)~7--12 years without high school diploma, (3) exactly 12 years with high school diploma, (4) $>12$ years without degree, (5)~Associate degree, (6)~Bachelor's degree, (7)~Professional or graduate degree & (1)~1--6 years, (2)~7--12 years, (3)~High school graduate, (4)~Some college, (5)~Associate degree, (6)~Bachelor's degree, (7)~Professional or graduate degree & $\left|x-y\right|$\\
        \bottomrule 
    \end{tabularx}
    \titlecaption{Coding of the demographic variables for the General Social Survey together with the feature maps for each variable.}{Seeds were provided with 16 options to choose from for their own ethnicity but only five options for their nominees. We attempt to unify the educational coding by combining the number of years of education and formal qualifications of the seeds to approximate the coding of nominees. The bias term in the first row of the table controls the overall edge density.\label{tbl:survey-gss-coding}}
\end{table}

\begin{table}
    \begin{tabularx}{\columnwidth}{lrrrr}
        \toprule 
        Dataset & Egos & Dropped egos & Alters & Dropped alters \\
        \midrule 
        GSS '04 & 2,774 & 38 (1.4\%) & 863 & 158 (15.5\%)\\
        ALP '09 & 2,472 & 0 (0.0\%) & 2,481 & 315 (11.3\%)\\
        BHPS '92 & 9,105 & 1 (\textless 0.1\%) & 18,219 & 506 (2.7\%)\\
        BHPS '94 & 8,728 & 5 (0.1\%) & 17,328 & 469 (2.6\%)\\
        BHPS '98 & 8,584 & 5 (0.1\%) & 16,949 & 565 (3.2\%)\\
        BHPS '00 & 8,281 & 2 (\textless 0.1\%) & 16,255 & 432 (2.6\%)\\
        BHPS '02 & 7,971 & 0 (0.0\%) & 15,716 & 502 (3.1\%)\\
        BHPS '04 & 7,609 & 0 (0.0\%) & 14,971 & 502 (3.2\%)\\
        BHPS '06 & 7,459 & 0 (0.0\%) & 14,558 & 331 (2.2\%)\\
        USS '11 & 36,526 & 199 (0.5\%) & 74,141 & 461 (0.6\%)\\
        USS '14 & 29,082 & 319 (1.1\%) & 61,892 & 611 (1.0\%)\\
        \bottomrule 
    \end{tabularx}
    \titlecaption{Number of retained seeds and nominees for each dataset together with the number of individuals who have been excluded from the analysis because one or more of their demographic attributes were missing.}{Individuals excluded for other reasons, e.g.\ due to being a relative or under the age of 18, are not listed.\label{tbl:survey-sample-size}}
\end{table}

\begin{table}
    \begin{tabularx}{\columnwidth}{lXXc}
        \toprule 
        Variable & Seed coding & Nominee coding & $f(x, y)$ \\
        \midrule 
        Bias term &\twocol{\dotfill}& $1$\\
        Age & Age in years recoded to match the nominee coding & Age brackets in years: (1)~0--20, (2)~21--35, (3)~36--50, (4)~51--65, (5)~66--80, (6)~$>80$ & $\left|x-y\right|$\\
        Sex & \twocol{(a)~Male, (b)~Female\dotfill} & $x\neq y$\\
        Ethnicity & \twocol{(a)~White or Caucasian, (b)~Black or African American, (c)~American Indian or Alaska Native, (d)~Asian or Pacific Islander, (e)~Hispanic (see below) (f)~Other} & $x \neq y$\\
        Hispanic & \twocol{(a)~Yes, (b)~No; ethnicity is coded as ``Hispanic'' if response is affirmative.\dotfill} & \\
        Education & \twocol{The seed coding is more refined but can be reduced to the nominee coding: (1)~Less than \ordinalnum{9} grade, (2)~\ordinalnum{9}--\ordinalnum{12} grade without diploma, (3)~High school graduate, (4)~Some college, (5)~Associate degree, (6)~Bachelor's degree, (7)~Master's degree, (8)~Professional degree or doctorate\dotfill} & $\left|x-y\right|$\\
        State & \twocol{One of 52 states and Washington DC and Puerto Rico\dotfill} & $x \neq y$\\
        \bottomrule 
    \end{tabularx}
    \titlecaption{Coding of the demographic variables for the American Life Panel together with the feature maps for each variable.}{We aggregate the ages and educational attainments of seeds to match the coarser coding of nominees, as shown in \cref{tbl:survey-alp-coding}. The joint effect of ethnic differences and whether people identify as Hispanic is still unclear~\cite{Smith2017a}; for consistency with the GSS, we code the ethnicity of respondents as ``Hispanic'' if they consider themselves to be Hispanic or Latino irrespective of their reported ethnicity. In fact, 46\% of respondents who identified as Hispanic selected ``other'' as their ethnicity, compared with $<1\%$ for respondents who did not identify as Hispanic.\label{tbl:survey-alp-coding}}
\end{table}

\begin{table}
    \noindent\begin{minipage}{\columnwidth}
    \renewcommand*\footnoterule{}
    \begin{tabularx}{\columnwidth}{lXXc}
        \toprule 
        Variable & Seed coding & Nominee coding & $f(x, y)$ \\
        \midrule 
        Bias term &\twocol{\dotfill}& $1$\\
        Age & \twocol{Age in years\dotfill} & $\left|x-y\right|$\\
        Sex & \twocol{(a)~Male, (b)~Female\dotfill} & $x\neq y$\\
        Occupation & (a)~\{Self-employed, employed, maternity leave, unpaid worker in family business\footnote{\label{ft:survey-only-usoc}Only available in Understanding Society.}\}, (b)~\{Unemployed, disabled\}, (c)~\{Full-time student, government training scheme\}, (d)~Family care, (e)~Retired & (a)~\{Full-time employed, part-time employed\}, (b)~Unemployed, (c)~Full-time education, (d)~Full-time housework, (e)~Retired & $x \neq y$\\
        Distance & Only applicable to the seed-nominee pair & (1)~$<1~\mathrm{mile}$, (2)~$<5~\mathrm{miles}$, (3)~$<50~\mathrm{miles}$, (4)~$\geq50~\mathrm{miles}$ but still in the UK &\footnote{We use the ordinal distance reported in the survey as a regression feature and generate control features using Monte Carlo simulation, as discussed in \cref{app:coding-and-features}.}\\
        Ethnicity\footref{ft:survey-only-usoc} & \twocol{Independent binary choices: White, Asian, Black, Other\dotfill} &\footnote{See \cref{app:coding-and-features} for a detailed description of the feature map.}\\
        \bottomrule 
    \end{tabularx}
    \end{minipage}
    \titlecaption{Coding of the demographic variables for the British Household Panel Survey and Understanding Society together with the feature maps for each variable.}{Sex and age have identical coding for seeds and nominees. We aggregate the detailed occupational coding of seeds to match the coding of nominees. In particular, we code women on maternity leave as employed because their occupational status is only temporary, and we code disabled individuals as ``not employed'' because they are unlikely to have the same social opportunities as people in employment.\label{tbl:survey-usoc-coding}}
\end{table}

\printbibliography

\end{refsection}

\end{document}